\documentclass[11pt]{article}
\usepackage{fullpage} 

\usepackage{amssymb,comment}
\usepackage{amsmath}
\usepackage{xspace}
\usepackage{tikz}
\usepackage{pgflibrarysnakes}
\usetikzlibrary{snakes}
\usepackage{times}
\usepackage{amstext}
\usepackage{calc}
\usepackage{amsopn}
\usepackage[noend]{algorithmic}
\usepackage[boxed]{algorithm}
\usepackage{eucal}
\usepackage{latexsym}
\usepackage{tabularx}

\usepackage{amsthm}

\newtheorem{theorem}{Theorem}
\newtheorem{corollary}[theorem]{Corollary}

\newtheorem{lemma}[theorem]{Lemma}

\theoremstyle{definition}
\newtheorem{definition}[theorem]{Definition}

{\bf}{\rm}

\newcommand{\cliquecover}{{\textsc{Edge Clique Cover}}\xspace}

\newcommand{\poly}{{\rm{poly}}}
\newcommand{\lnV}{\ell}
\newcommand{\ilnV}{\gamma}
\newcommand{\mlnV}{\delta}
\newcommand{\cstr}{\mathbf{c}}
\newcommand{\ccfam}{\ensuremath{\mathcal{C}}}
\newcommand{\Efree}{\ensuremath{E^{\textrm{free}}}}
\newcommand{\Gfree}{\ensuremath{G^{\textrm{free}}}}
\newcommand{\Eimp}{\ensuremath{E^{\textrm{imp}}}}
\newcommand{\Vars}{\mathtt{Vars}}
\newcommand{\ivar}{i}
\newcommand{\icl}{j}
\newcommand{\ibit}{c}
\newcommand{\icopy}{\eta}
\newcommand{\clause}{\Psi}
\newcommand{\ilit}{\alpha}
\newcommand{\iend}{\beta}

\newcommand{\compass}{$\textrm{NP} \subseteq \textrm{coNP}/\textrm{poly}$}

\newcommand{\defproblemu}[3]{
  \vspace{1mm}
\noindent\fbox{
  \begin{minipage}{0.95\textwidth}
  #1 \\
  {\bf{Input:}} #2  \\
  {\bf{Question:}} #3
  \end{minipage}
  }
  \vspace{1mm}
}
\newcommand{\Oh}{\ensuremath{O}}

\begin{document}

  \date{}

\author{
  Marek Cygan 
  \thanks{
    IDSIA, University of Lugano, Switzerland,
      \texttt{marek@idsia.ch}
  }
  \and
  Marcin Pilipczuk
  \thanks{
    Institute of Informatics, University of Warsaw, Poland,
      \texttt{malcin@mimuw.edu.pl}
  }
  \and
  Micha\l{} Pilipczuk 
  \thanks{
    Department of Informatics, University of Bergen, Norway, \texttt{michal.pilipczuk@ii.uib.no}
  }
}

  \title{Known algorithms for \cliquecover{} are probably optimal}

\maketitle

\begin{abstract}
In the \cliquecover{} (ECC) problem, given a graph $G$ and an integer $k$, we ask whether the edges of $G$ 
can be covered with $k$ complete subgraphs of $G$ or, equivalently, 
whether $G$ admits an intersection model on $k$-element universe. 
Gramm et al. [JEA 2008] have shown a set of simple rules that reduce the number of vertices of $G$ to $2^k$, and no algorithm is known with significantly better running time bound than a brute-force search on this reduced instance. In this paper we show that the approach of Gramm et al. is essentially optimal: we present a polynomial time algorithm that reduces an arbitrary $3$-CNF-SAT formula with $n$ variables and $m$ clauses to an equivalent ECC instance $(G,k)$ with $k = O(\log n)$ and $|V(G)| = O(n + m)$. Consequently, there is no $2^{2^{o(k)}}{\rm poly}(n)$ time algorithm for the ECC problem, unless the Exponential Time Hypothesis fails. To the best of our knowledge, these are the first results for a natural, fixed-parameter tractable problem, and proving that a doubly-exponential dependency on the parameter is essentially necessary.
\end{abstract}

\section{Introduction}

Recently, there has been an increasing interest in not only
improving the running times of exact algorithms for various NP-hard problems,
but also in finding the limits of such improvements.
{\em{Parameterized complexity}} is a very useful framework for such study:
in this approach, an instance~$x$ of a 
parameterized problem comes with an integer parameter~$k$.
We say that a problem is {\em{fixed parameter tractable}} ({\em{FPT}}) if there exists an algorithm solving any instance~$(x,k)$ in
time~$f(k) |x|^{\Oh(1)}$ for some computable function~$f$.
In other words,
 the exponential explosion of the running time, probably unavoidable for NP-hard problems,
is encapsulated in a function of a parameter.
With a wide range of possible parameters (e.g., intended solution size or structural
 graph parameters), the parameterized complexity paradigm allows much deeper insight
into the hardness of NP-hard problems than the classical instance-size measure.

Although the definition of a fixed-parameter algorithm
  allows an arbitrarily fast growing function $f$,
if $f$ turns out to be relatively small (say, single-exponential),
a fixed-parameter algorithm becomes practical for a reasonable range
of values of the parameter.
Therefore, since the dawn of parameterized complexity, 
researchers try to bound, as much as possible, the explosion of the running
time hidden in the function $f$.
In particular the goal of the part of parameterized complexity
called by Marx~\cite{marx-future} the {\em optimality program},
is to quantitatively understand what is the best possible $f$ function in the running time.

In the last few years, this trend has been complemented by a research of lower bounds,
usually conditional on the {\em{Exponential Time Hypothesis}} by Impagliazzo et al.
\cite{ImpagliazzoPZ01}. Let $c_k$ be the infimum over all positive
reals $c$ such that there exists an algorithm resolving satisfiability
of $n$-variable $k$-CNF-SAT formulae in $\Oh(2^{cn})$ time. The Exponential Time
Hypothesis (ETH) asserts that $c_3>0$ (that is, $3$-CNF-SAT formulae cannot
    be resolved in subexponential time in the number of variables), whereas
its stronger variant, {\em{Strong Exponential Time Hypothesis}} (SETH) \cite{seth-def,seth-first},
   asserts
that $\lim_{k \to \infty} c_k = 1$ (in particular, an exhaustive search
is the best possible solution for an arbitrary CNF formula).

Since the seminal paper of Impagliazzo et al. \cite{ImpagliazzoPZ01},
many ETH-based lower bounds were developed (e.g., \cite{fomin-eth,marx-eth1,marx-eth2}),
 with the prominent example
of tight bounds for $W[1]$-hard problems \cite{chkx2,chkx1}.
The most standard usage of ETH is to refute
an existence of a subexponential algorithm for some problem
by showing a linear (in the number of variables and clauses)
  reduction from an arbitrary 3-CNF-SAT formula and
  using to the so-called {\em{Sparsification Lemma}} \cite{ImpagliazzoPZ01}.
In order to prove lower bounds for different complexities (different functions $f$, say
$f(k) = 2^{\Oh(k \log k)}$), the allowed dependency of $k$ on $n$ and $m$ is more involved.
Moreover, note that an ETH-based lower bound may exclude a subexponential algorithm,
but says nothing about the possible base of the exponent in an $\Oh(c^k |x|^{\Oh(1)})$
fixed-parameter algorithm.
For such lower bounds, one needs to assume SETH
and either make the parameter independent of the number of clauses,
or carefully use {\em{Sparsification Lemma}} as in some reductions of~\cite{seth-reductions}.

In 2011, two important results of by Lokshtanov, Marx and Saurabh \cite{lms:treewidth,lms:k-to-k}
show how to overcome the aforementioned difficulties:
they developed a lower bound framework for proving ETH-based slightly superexponential lower bounds
\cite{lms:k-to-k} as well as show that some core dynamic programming routines
on graphs of bounded treewidth are optimal assuming SETH \cite{lms:treewidth}.
A significant contribution of these results is the development of new, sophisticated gadgets, which were found to be very inspiring.
For example, both research directions initiated by Lokshtanov, Marx and Saurabh lead to settling down the tight
bounds for algorithms for connectivity problems on graphs of bounded treewidth
\cite{cut-and-count}.

In this paper we make one step further in the quest of finding tight runtime bounds
for parameterized problems by presenting (to the best of our knowledge) the first known double-exponential lower bound.
Our problem of interest is an important combinatorial problem called \cliquecover{}.

\defproblemu{\cliquecover}{An undirected graph~$G$ and an integer~$k$.}
{Does there exist a set of~$k$ subgraphs of~$G$, such that
each subgraph is a clique and each edge of~$G$ is contained in at least one of these subgraphs?}

\cliquecover{} is known to be NP-complete even in very restricted graph classes~\cite{chang-muller:cc,hoover:cc,orlin:cc}
and was widely studied under a few different names: \textsc{Covering by Cliques} (GT17), \textsc{Intersection Graph Basis} (GT59)~\cite{garey-johnson} and \textsc{Keyword Conflict}~\cite{kellerman}.
It is known that the \cliquecover{} problem is equivalent to the problem of finding a representation
of a graph~$G$ as an intersection model with at most~$k$ elements in the universe \cite{cliquecover:erdos,cliquecover:book,cliquecover:appl}.
Therefore a covering of a complex real-world network by a small number of cliques may reveal its hidden structure~\cite{guillaume:cc}.
Further applications of \cliquecover{} can be found in such various areas as 
  computational geometry~\cite{cc-apl1},
  applied statistics~\cite{cc-apl2,cc-apl3}, and compiler optimization~\cite{cc-apl4}.
Due to its importance, the \cliquecover{} problem was studied from various perspectives,
including approximation upper and lower bounds~\cite{apx:cc,lund-yannakakis},
heuristics~\cite{bt:cc,cc-apl2,kellerman,kou:cc,cc-apl3,cc-apl4} and
polynomial-time algorithms for special graph classes~\cite{hoover:cc,cc-class2,cc-class1,orlin:cc}.

From the point of view of exact algorithms, a natural parameterization of \cliquecover{} by the number of cliques
was studied by Gramm et al.~\cite{gghn:cc}. The authors propose a set of simple rules
that reduce the number of vertices of the input graph to $2^k$ (the so-called {\em{kernel}}).
Currently the best known fixed-parameter algorithm for \cliquecover{} parameterized by $k$ is
a brute-force search on the $2^k$-vertex kernel, which runs in double-exponential time in terms of $k$.
Due to the importance of the \cliquecover{} problem on one hand,
    and the lack of any improvement upon
the very simple approach of Gramm et al.~\cite{gghn:cc} on the other hand,
    \cliquecover{} became a seasoned veteran of open problem sessions
(with the most recent occurrence on the last Workshop on Kernels in Vienna, 2011).
Only very recently, a superset of the current authors \cite{destylacje}
have shown that \cliquecover{} is compositional, refuting
(under the assumption $\textrm{NP} \not\subseteq \textrm{coNP}/\textrm{poly}$)
  an existence of a polynomial-time
algorithm that reduces the size of the input graph to polynomial in $k$
(the so-called {\em{polynomial kernel}}).

In this paper we complete the picture of the parameterized complexity of \cliquecover{}.
Our main technical result is the following reduction.
\begin{theorem}\label{thm:eth-reduction}
 There exists a polynomial-time algorithm that, given a 3-CNF-SAT formula with $n$ variables and $m$ clauses,
 constructs an equivalent \cliquecover{} instance $(G,k)$ with $k = \Oh(\log n)$ and
 $|V(G)| = \Oh(n+m)$.
\end{theorem}
The above theorem, together with a well-known result that
an existence of a subexponential (in the number of variables
and clauses) algorithm for verifying satisfiability of $3$-CNF formulae
violates ETH \cite{ImpagliazzoPZ01}, proves the following lower bound.
\begin{corollary}
Unless ETH fails, there does not exist an algorithm solving \cliquecover{}
in time $\Oh(2^{2^{o(k)}} |V(G)|^{\Oh(1)})$.
\end{corollary}
We note that ETH is not necessary to refute existence of single-exponential
algorithms for the problem.
\begin{corollary}
Unless all problems in NP are solvable in quasipolynomial time,
there does not exist an algorithm solving \cliquecover{}
in time $\Oh(2^{k^{\Oh(1)}} |V(G)|^{\Oh(1)})$.
\end{corollary}
Note that \cliquecover, as a covering problem, can be solved
by a dynamic programming algorithm in time $2^{\Oh(|E(G)| + |V(G)|)}$.
Thus, our result shows also the kernelization hardness for \cliquecover{}.
\begin{corollary}
Unless ETH fails, there does not exist a polynomial-time algorithm that,
given an \cliquecover{} instance $(G,k)$, outputs an equivalent instance $(G',k')$
with $|V(G')| + |E(G')|$ bounded by $2^{o(k)}$.
\end{corollary}
Moreover, note that if we would like to refute a polynomial kernel for \cliquecover{},
   we need significantly weaker assumption than ETH.

\begin{corollary}\label{cor:no-poly-kernel}
Unless all problems in NP are solvable in quasipolynomial time,
there does not exist a polynomial-time algorithm that,
given an \cliquecover{} instance $(G,k)$, outputs an equivalent instance $(G',k')$
with $|V(G')| + |E(G')|$ bounded polynomially in $k$.
\end{corollary}

Finally, our reduction shows, that a polynomial time preprocessing routine
reducing any instance of $(G,k)$ of \cliquecover{} to size $2^{o(k)}$,
would make it possible to compress any $r$-CNF-SAT instance to $n^{o(r)}$
bits, which is not possible unless \compass, as observed by Dell and van Melkebeek~\cite{dell:kernel}.
Even more generally, we can rule out polynomial time compression to any language $L$
with the output having size $2^{o(k)}$.

\begin{corollary}\label{cor:compass}
Unless \compass, there does not exist a polynomial-time algorithm that,
given an \cliquecover{} instance $(G,k)$, outputs an equivalent instance $A(G,k) \in \Sigma^*$
of some language $L$ (i.e. $A(G,k) \in L$ iff $(G,k)$ is a YES-instance)
with $|A(G,k)| = 2^{o(k)}$.
\end{corollary}

\begin{proof}
Consider a formula of $r$-CNF-SAT $\phi$ with $n$ variables and $m$ clauses.
W.l.o.g. we may assume, that $m=O((2n)^r)$, since otherwise we may remove
repeated clauses.
Our first step is the standard reduction from $r$-CNF-SAT to $3$-CNF-SAT,
which in polynomial time outputs a formula $\phi'$ with $O(rm)=O(r(2n)^r)$ variables
and $O(rm)$ clauses.
Next, we use our reduction from Theorem~\ref{thm:eth-reduction}, to
create an instance $(G,k)$ of \cliquecover{} with $k=O(r\log n)$
and $|V(G)|=O(rm)$.
Observe, that having a compression routine for \cliquecover{} 
with $|A(G,k)|$ bounded by $2^{o(k)}$ would make
it possible to express the initial formula $\phi$
with $n^{o(r)}$ bits, which is not possible due to Dell and van Melkebeek~\cite{dell:kernel}.
\end{proof}

To the best of our knowledge, the assumption $\textrm{NP} \not\subseteq \textrm{QP}$
and the assumption $\textrm{NP} \not\subseteq \textrm{coNP}/\textrm{poly}$,
used (apart from Corollary \ref{cor:compass}) to refute a polynomial kernel for \cliquecover{} in \cite{destylacje}, are not known to be comparable.
We are not aware of any other than Corollary \ref{cor:no-poly-kernel} polynomial kernelization hardness result for fixed-parameter tractable problem
outside the framework of compositions (for more on compositions, see e.g.
\cite{bodlaender:kernel,cross-composition,dell:kernel,fortnow:kernel}).

Throughout the paper we investigate the graph we denote as $H_\lnV$, 
which is isomorphic to a clique on $2^\lnV$ vertices with a perfect matching removed,
called the cocktail party graph.
The core idea of the proof of Theorem \ref{thm:eth-reduction} is the observation
that a graph $H_\lnV$ is a hard instance for the \cliquecover{} problem, at least from the point of view
of the currently known algorithms.
Such a graph, while being immune to the reductions of Gramm et al. \cite{gghn:cc},
can be quite easily covered with $2\lnV$ cliques, and there are multiple solutions
of such size. Moreover, it is non-trivial to construct smaller clique covers for $H_\lnV$ (but they exist).

In fact, the optimum size of a clique cover of cocktail party graphs with $2n$ vertices
is proved to be $\min(k: n \le \binom{k-1}{\lceil k/2 \rceil})$ for all $n > 1$
by Gregory and Pullman~\cite{gregory-pullman}.
Moreover Chang et al. study cocktail party graphs in terms of rankwidth,
which they prove to be unbounded in case of edge clique graphs of cocktail party graphs~\cite{chang} (we refer to their work for appropriate definitions).

\paragraph{Acknowledgements.} We thank Leszek Ko\l{}odziejczyk for some enlightening
discussions on the complexity assumptions mentioned in this paper.
Moreover we thank Daniel Lokshtanov and Ton Kloks for helpful comments and discussions.

\section{Double-exponential lower bound}

This section is devoted to the proof of Theorem \ref{thm:eth-reduction}.
We start by introducing some notation.
For an undirected graph $G$, by $V(G)$ and $E(G)$ we denote its vertex-
and edge-sets respectively.
For a set $X \subseteq V(G)$, a subgraph induced by $X$ is denoted by $G[X]$.
For two sets $X,Y \subseteq V(G)$, the set $E(X,Y)$ contains all edges of $G$
that have one endpoint in $X$ and a second endpoint in $Y$.
As in this section we talk mostly about cliques,
   we allow ourselves some shortcuts in notation.
If $X$ is a subgraph or a subset of vertices of $G$, we call
$X$ {\em{a clique in}} $G$ if $X$ or $G[X]$ is a complete graph.
We also often identify a subgraph being a clique in $G$ with its vertex
set.
Moreover, for a bit-string $\mathbf{c}$, by $\overline{\mathbf{c}}$ we denote
its bitwise negation.

Recall that a graph $H_\lnV$ is defined as a clique on $2^\lnV$ vertices
with a perfect matching removed.
In Section \ref{ss:eth-clique} we analyze in details graphs $H_\lnV$.
It turns out that there is a large family of clique covers of size $2\lnV$, where
the clique cover consists of pairs of cliques, whose vertex sets are complements.
We refer to such pairs as to {\em{clique twins}} and a clique cover consisting of clique twins
is a {\em{twin clique cover}}. In particular,
given any clique $C$ of size $2^{\lnV-1}$ in the graph $H_\lnV$, we can construct
a twin clique cover of $H_\lnV$ that contains $C$.
Note that we have $2^{2^{\lnV-1}}$ such starting cliques $C$ in the graph $H_\lnV$: for each edge
of the removed perfect matching, we choose exactly one endpoint into the clique.
In our construction $\lnV = \theta(\log n)$ and
this clique $C$ encodes the assignment of the variables in the input $3$-CNF-SAT formula.

Section \ref{ss:eth-construction} contains the details of our construction.
We construct a graph $G$ with edge set $E(G)$ partitioned into two sets
$\Eimp$ and $\Efree$. The first set contains the important edges, that is, the ones that are nontrivial to cover
and play important role in the construction. The second set contains edges that are covered
for free; in the end we add simplicial vertices (i.e., with neighbourhood being a clique) to the graph $G$
to cover $\Efree$. Note that without loss of generality we can assume that a closed neighbourhood
of a non-isolated simplicial vertex is included in any optimal clique cover; however, we need to take care
that we do not cover any edge of $\Eimp$ with such a clique and that we use only $\Oh(\log n)$ simplicial vertices
(as each such vertex adds a clique to the solution).

While presenting the construction in Section \ref{ss:eth-construction}, we give informal explanations of the
role of each gadget. In Section \ref{ss:eth-completeness} we show formally how to translate a satisfying assignment of the input formula
into a clique cover of the constructed graph, whereas the reverse translation is provided in Section \ref{ss:eth-soundness}.

\subsection{Cocktail party graph}\label{ss:eth-clique}

In this section we analyze graphs $H_\lnV$ known as cocktail party graphs; see Figure~\ref{fig:clique-without-matching} for an illustration for small values of $\lnV$. Recall that for an integer $\lnV \geq 1$
the graph $H_\lnV$ is defined as a complete graph on $2^\lnV$ vertices with a perfect matching removed.
Note that a maximum clique in $H_\lnV$ has $2^{\lnV-1}$ vertices (i.e., half of all the vertices),
and contains exactly one endpoint of each non-edge of $H_\lnV$. Moreover, if $C$ is a maximum clique
in $H_\lnV$, so is its complement $V(H_\lnV) \setminus C$. This motivates the following definition.

\begin{figure}[htb]
\begin{center}
\includegraphics{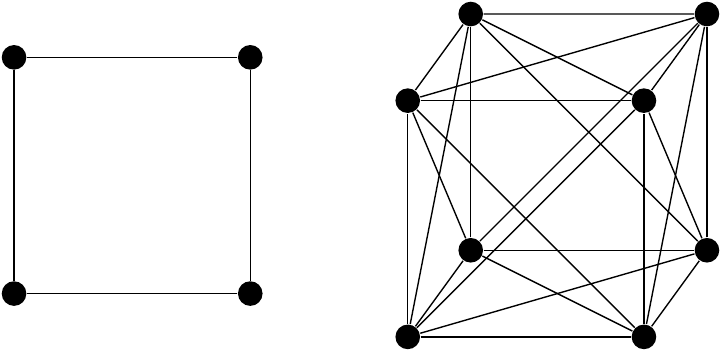}
\caption{The graphs $H_\lnV$ for $\lnV=2$ and $\lnV=3$.
  In the first case the optimum clique cover contains four two-vertex cliques and is a twin clique cover.
  In the second case an example twin clique cover is the set of all six faces of the cube; however,
     there exists a non-twin clique cover of $H_3$ of size five.}
\label{fig:clique-without-matching}
\end{center}
\end{figure}

\begin{definition}\label{def:twins}
A pair of maximum cliques $C$ and $V(H_\lnV) \setminus C$ in $H_\lnV$ is called {\em{clique twins}}.
A clique cover of $H_\lnV$ that consists of clique twins is called a {\em{twin clique cover}}.
\end{definition}

The following lemma describes structure of twin clique covers of $H_\lnV$ of size $2\lnV$ (i.e., containing
$\lnV$ twins).

\begin{lemma}\label{lem:extend-twin}
Assume we are given a set $\ccfam_0$ of $1 \leq \mlnV \leq \lnV$ clique twins with the following property:
if we choose one clique from each of the $\mlnV$ clique twins, the intersection of the vertex sets of these cliques has size exactly $2^{\lnV-\mlnV}$.
Then there exists a twin clique cover $\ccfam$ of size $2\lnV$ that contains $\ccfam_0$.
\end{lemma}
\begin{proof}
We arbitrarily number the clique twins from $\ccfam_0$ with numbers from $1$ to $\mlnV$, and in each clique twin we distinguish
one clique labeled $0$ and one labeled $1$. Then, to each vertex $v \in V(H_\lnV)$, we assign a $\mlnV$-bit string
$\cstr_v$ that on a position $\ilnV$ ($1 \leq \ilnV \leq \mlnV$) contains the bit assigned to the clique that contains $v$
from the $\ilnV$-th clique twins. Note that, as all subgraphs in $\ccfam_0$ are cliques, for any non-edge $uv$ of $H_\lnV$,
the strings $\cstr_u$ and $\cstr_v$ are bitwise negations.

Fix a $\mlnV$-bit string $\cstr$ and consider $X_\cstr = \{v \in V(H_\lnV): \cstr_v = \cstr\}$.
Note that any clique in $\ccfam_0$ contains entire $X_\cstr$ or entire $X_{\overline{\cstr}}$ and,
by the assumptions of the lemma, $|X_\cstr| = 2^{\lnV-\mlnV}$.
Moreover, as for a non-edge $uv$ the strings $\cstr_u$ and $\cstr_v$ are bitwise negations,
  each non-edge of $H_\lnV$ connects a vertex from $X_\cstr$, for some $\cstr$,
  with a vertex from $X_{\overline{\cstr}}$.
We can now label each vertex $v \in X_\cstr$ with bit string $\cstr_v'$ of length $(\lnV-\mlnV)$, such that
in $X_\cstr$ each vertex receives a different label, and if $uv$ is a non-edge of $H_\lnV$, then
$\cstr_u' = \overline{\cstr_v'}$. In this manner each vertex $v \in V(H_\lnV)$ receives a unique
$\lnV$-bit label $\cstr_v\cstr_v'$ and for any non-edge $uv$ of $H_\lnV$ we have $\cstr_u\cstr_u' = \overline{\cstr_v\cstr_v'}$.

For an integer $1 \leq \ilnV \leq \lnV$ and a bit $c \in \{0,1\}$, consider a set $C_{\ilnV,c}$
consisting of those vertices of $H_\lnV$ whose aforementioned $\lnV$-bit labels have $\ilnV$-th bit set to $c$.
As in $H_\lnV$ a vertex $v$ is connected with all other vertices except the one labeled with the bitwise negation
of the label of $v$, $C_{\ilnV,c}$ induces a clique.
Moreover, for any edge $uv \in E(H_\lnV)$, the labels of $u$ and $v$ agree on at least one bit,
  and the corresponding clique $C_{\ilnV,c}$ contains the edge $uv$.
As $C_{\ilnV,0} = V(H_\lnV) \setminus C_{\ilnV,1}$, we infer that
the family $\ccfam = \{C_{\ilnV,c}: 1 \leq \ilnV \leq \lnV, c \in \{0,1\}\}$ is a twin clique cover of $H_\lnV$
of size $2\lnV$. We finish the proof of the lemma by noticing that $\{C_{\ilnV,c} : 1 \leq \ilnV \leq \mlnV, c \in \{0,1\}\} = \ccfam_0$.
\end{proof}

Note that the above lemma for $\mlnV=1$ implies that for any maximum clique $C$ in $H_\lnV$ there exists a twin clique cover of size $2\lnV$ that contains $C$. The next lemma treats about optimum twin clique covers of $H_\lnV$.

\begin{lemma}\label{lem:last-twin}
Let $\ccfam$ be a clique cover of $H_\lnV$ that contains at least $\lnV-1$ clique twins.
Then $|\ccfam| \geq 2\lnV$ and, if $|\ccfam| = 2\lnV$, then $\ccfam$ is a twin clique cover of $H_\lnV$.
\end{lemma}
\begin{proof}
Let $\ccfam_0 \subseteq \ccfam$ be a set of $2\lnV-2$ cliques that form the assumed $\lnV-1$ clique twins.
We use the family $\ccfam_0$ to label the vertices of $H_\lnV$ with $(\lnV-1)$-bit labels similarly as
in the proof of Lemma \ref{lem:extend-twin}.
That is, we arbitrarily number these clique twins with numbers $1$ to $\lnV-1$, and in each clique twin we distinguish
one clique labeled $0$ and one labeled $1$; a string $\cstr_v$ for $v \in V(H_\lnV)$ consists of $(\lnV-1)$
bits assigned to the cliques containing $v$. Again, for any non-edge $uv$ of $H_\lnV$, the strings $\cstr_u$ and $\cstr_v$ are bitwise negations.

Fix a $(\lnV-1)$-bit string $\cstr$ and consider $X_\cstr = \{v \in V(H_\lnV): \cstr_v = \cstr\}$. Note
that any clique in $\ccfam_0$ contains entire $X_\cstr$ or entire $X_{\overline{\cstr}}$ and thus
no clique in $\ccfam_0$ covers the edges of $E(X_\cstr, X_{\overline{\cstr}})$.
Moreover, as for any non-edge $uv$ we have $\cstr_u = \overline{\cstr_v}$, the sets $X_\cstr$
and $X_{\overline{\cstr}}$ are of equal size.

As $\ccfam$ is a clique cover of $H_\lnV$, $\ccfam \setminus \ccfam_0$ covers $E(X_\cstr, X_{\overline{\cstr}})$. As
$H_\lnV[X_\cstr \cup X_{\overline{\cstr}}]$ is isomorphic to $K_{2|X_\cstr|}$ with a perfect matching removed,
a direct check shows that if $|X_\cstr| \geq 3$ then $|\ccfam \setminus \ccfam_0| \geq 3$, that is,
  we need at least three cliques to cover $E(X_\cstr, X_{\overline{\cstr}})$.
Thus, if $|\ccfam| \leq 2\lnV$, for each string $\cstr$ we have $|X_\cstr| \leq 2$.
As there are $2^{\lnV-1}$ bit strings $\cstr$, and $2^\lnV$ vertices of $H_\lnV$,
we infer that in this case $|X_\cstr| = 2$ for each bit string $\cstr$.

Fix a bit string $\cstr$. If $|X_\cstr| = 2$, then the graph $H_\lnV[X_\cstr \cup X_{\overline{\cstr}}]$ is isomorphic
to a $4$-cycle and $E(X_\cstr, X_{\overline{\cstr}})$ contains two opposite edges of this cycle.
These edges cannot be covered with a single clique. We infer that $|\ccfam \setminus \ccfam_0| \geq 2$, i.e.,
$|\ccfam| \geq 2\lnV$. Assume now that $|\ccfam| = 2\lnV$ and let $\ccfam \setminus \ccfam_0 = \{C,C'\}$.
Note that for any bit string $\cstr$ the clique $C$ contains both endpoints
of one edge of $E(X_\cstr,X_{\overline{\cstr}})$, and $C'$ contains the endpoints of the second edge.
Therefore $C = V(H_\lnV) \setminus C'$
and the lemma is proven.
\end{proof}

Let us remark that Lemma~\ref{lem:last-twin} implies that one cannot cover the graph $H_\ell$ with less than $\ell$ clique twins, i.e., the bound given by Lemma~\ref{lem:extend-twin}. Indeed, assume that there exists a cover of $H_\ell$ using $\ell'<\ell$ clique twins. If necessary, copy some of the clique twins in order to obtain a cover that uses exactly $\ell-1$ twins. However, from Lemma~\ref{lem:last-twin} we infer that this cover needs to contain in fact more cliques, a contradiction.

\subsection{Construction}\label{ss:eth-construction}

Recall that, given a $3$-CNF-SAT formula $\Phi$, we are to construct an equivalent \cliquecover{} instance with the number of cliques bounded
logarithmically in the number of variables of $\Phi$.
We start with an empty graph $G$, and we subsequently add new gadgets to $G$. Recall that the edge set of $G$ is partitioned into $\Efree$ and $\Eimp$;
at the end of this section we show how to cover the set $\Efree$ with a small number of cliques, each induced by a closed neighbourhood of a simplicial vertex.
We refer to Figure \ref{fig:construction} for an illustration of the construction.
\begin{figure}
\begin{center}
\includegraphics{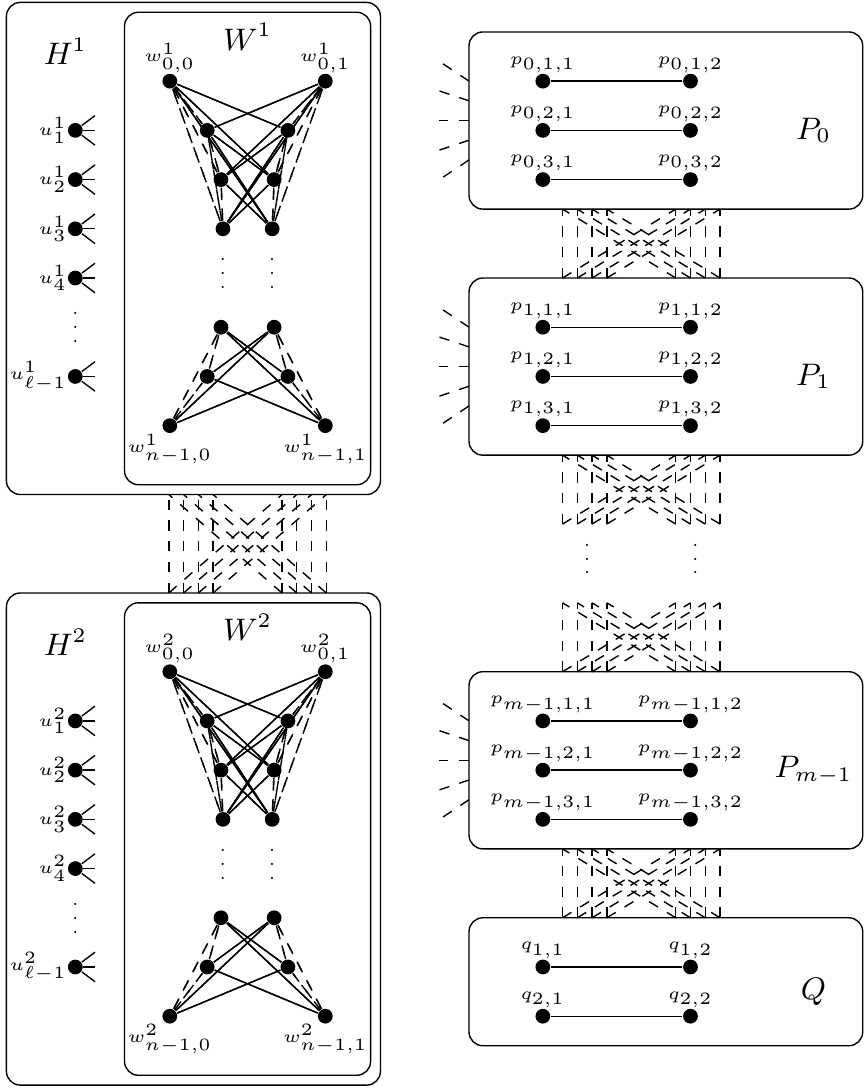}
\caption{Illustration of the construction of the graph $G$. Solid edges belong to the set $\Eimp$ and dashed edges to $\Efree$.
  The simplicial vertices are not illustrated, nor are the details which edges between
    the clause gadgets $P_\icl$ and the assignment gadgets $H^\icopy$ are present in the graph $G$.
    Moreover not all edges of $\Efree$ between the vertices of gadgets $Q$ and $P_j$ are shown.
\vspace{-0.8cm}}
  \label{fig:construction}
  \end{center}
  \end{figure}

\subsubsection{Preprocessing of the formula $\Phi$} Let $\Vars$ denote the set of variables of $\Phi$.
By standard arguments, we can assume that in $\Phi$ each clause consists of exactly $3$ literals, these literals
contain different variables and no clause appears more than once.
Moreover, we perform the following two regularization operations on $\Phi$. First, we introduce some dummy variables into $\Vars$, the set of variables of $\Phi$, so that
the number of variables is a power of two, and that there exists at least one dummy variable (i.e., the variable that does not appear in any clause). This operation at most
doubles the number of variables in $\Vars$. Second, we ensure that if $\Phi$ is satisfiable, then there exists a satisfying assignment of $\Phi$ that assigns true to exactly half of the
variables, and false to the other half. This can be done by transforming $\Phi$ into $\Phi'=\Phi\wedge \overline{\Phi}$, where $\overline{\Phi}$ is a copy of $\Phi$ on duplicated set of variables $\overline{\Vars}$ and, moreover, all the literals in $\overline{\Phi}$ are replaced with their negations. Clearly, if $\Phi'$ is satisfiable then the same assignment satisfies $\Phi$ in particular. Moreover, note that any satisfying assignment of $\Phi$ can be extended to a satisfying assignment of $\Phi'$ by assigning each copy of a variable the negation of the value of the original; this assignment assigns true to exactly half of the variables, and false to the other half.
Observe that the second operation does not change the properties ensured by the first operation, as it exactly doubles the number of variables. Moreover, in the satisfying assignment we can fix value of one dummy variable in $\Vars$.

After performing the described operations, let $n$ be the number of variables in $\Vars$, $m$ be the number of clauses of $\Phi$, and $n = 2^\lnV$. Note that $m = \Oh(n^3)$ and $\log m = \Oh(\log n) = \Oh(\lnV)$.

\subsubsection{Assignment-encoding gadget}
We assume that $\Vars = \{x_0,x_1,\ldots,x_{n-1}\}$ and that the $0$-th variable is a dummy one (it serves in the construction as a true pattern).
Take a graph $H$ isomorphic to $H_{\lnV+1}$ and denote its vertices by $w_{\ivar,\ibit}$ for $0 \leq \ivar < n$, $\ibit \in \{0,1\}$;
the non-edges of $H_{\lnV+1}$ are $\{w_{\ivar,0}w_{\ivar,1}: 0 \leq \ivar < n\}$.
Let $W_\ibit = \{w_{\ivar,\ibit}: 0 \leq \ivar < n\}$ and $W = W_0 \cup W_1$. We put the edges of $H[W_0]$ and $H[W_1]$ into $\Efree$ and $E(W_0,W_1)$ into $\Eimp$.

Moreover, we add $(\lnV-1)$ vertices $u_\ilnV$, $1 \leq \ilnV < \lnV$, to the graph $H$. Each vertex $u_\ilnV$ is connected to all vertices of $W$ via edges belonging to $\Eimp$.
This finishes the description of the assignment-encoding gadget $H$. We add two copies of the gadget $H$ to the graph $G$, and denote the vertices of the $\icopy$-th copy ($\icopy \in \{1,2\}$)
by $w_{\ivar,\ibit}^\icopy$ and $u_\ilnV^\icopy$. We define $W^\icopy$, $W_0^\icopy$, $W_1^\icopy$ and $H^\icopy$ in the natural way.
In the graph $G$, for all indices $(\ivar,\ibit,\ivar',\ibit')\in (\{0,1,\ldots,n-1\}\times \{0,1\})^2$ we connect each pair of vertices $w_{\ivar,\ibit}^1$ and $w_{\ivar',\ibit'}^2$ with an edge from $\Efree$, i.e., we introduce a complete bipartite graph with edges belonging to $\Efree$ between $W^1$ and $W^2$.

Let us now describe the intuition behind this construction. In the gadget $H$, the neighbourhood of each vertex $u_\ilnV$ is not a clique, thus any clique cover
of $H$ needs to include at least two cliques that contain $u_\ilnV$. However, if we are allowed to use only two cliques per vertex $u_\ilnV$,
these cliques need to induce clique twins in the subgraph $H_{\lnV+1}$ of $H$. With the assumption of only two cliques per vertex $u_\ilnV$,
the set of $(\lnV-1)$ vertices $u_\ilnV$ ensures that when covering $H_{\lnV+1}$ we use $\lnV$ clique twins: $(\lnV-1)$ from the vertices $u_\ilnV$ and one given by the edges in $\Efree$ (cliques $W_0$ and $W_1$).
Lemma \ref{lem:last-twin} asserts that the optimal way to complete a clique cover of $H_{\lnV+1}$ is to use one more pair of clique twins: this clique twins, called {\em{the assignment
 clique twins}}, encode the assignment
(and, as they are not bounded by the vertices $u_\ilnV$, they can be used to verify the assignment in the forthcoming gadgets).
Finally, we need two copies of the gadget $H$, as in the soundness proof we have one free clique that spoils the aforementioned argument; however, as the vertices $\{u_\ilnV^\icopy: 1 \leq \ilnV < \lnV, \icopy \in \{1,2\}\}$ form an independent set, it cannot spoil it in both copies at once. The edges between the sets $W$ in the copies allow us to use the same
two cliques as the assignment clique twins in both copies of the gadget $H$.

One could ask why we put edges from $H[W_0]$ and $H[W_1]$ into $\Efree$, since in the previous section we have assumed that
all the edges of $H_{\lnV}$ are to be covered.
The reason for this is that additional cliques with several vertices from $W_0$ or $W_1$ will appear,
in order to cover other edges of $\Efree$ and for this reason we need to put edges from $H[W_0]$ and $H[W_1]$ into $\Efree$
and carefully investigate cliques that cover those edges.


\subsubsection{Clause gadgets}
We now introduce gadgets that verify correctness of the assignment encoded by the assignment clique twins, described in the previous paragraphs.

First, let us extend our notation.
Let $\Phi = \clause_0 \wedge \clause_1 \wedge \ldots \wedge \clause_{m-1}$ and for integers $\icl,\ilit$ let $\ivar(\icl,\ilit)$ be the index of the variable that appears in $\ilit$-th literal
in the clause $\clause_\icl$. Moreover, let $\ibit(\icl,\ilit) = 0$ if the $\ilit$-th literal of $\clause_\icl$ is negative (i.e., $\neg x_{\ivar(\icl,\ilit)}$)
  and $\ibit(\icl,\ilit)=1$ otherwise.

For each clause $\clause_\icl$, we introduce into $G$ a subgraph $P_\icl$ isomorphic to $3K_2$, that is, $V(P_\icl) = \{p_{\icl,\ilit,\iend} : 1 \leq \ilit \leq 3, \iend \in \{1,2\}\}$
and $E(P_\icl) = \{p_{\icl,\ilit,1}p_{\icl,\ilit,2}: 1 \leq \ilit \leq 3\}$. Moreover, we introduce into $G$ a guard subgraph $Q$ isomorphic to $2K_2$, that is,
$V(Q) = \{q_{1,1}, q_{1,2}, q_{2,1}, q_{2,2}\}$ and $E(Q) = \{q_{1,1}q_{1,2}, q_{2,1}q_{2,2}\}$.

All the edges in all gadgets $P_\icl$ and $Q$ belong to $\Eimp$. Moreover, we introduce the following edges to $\Efree$.
First, for each vertex of $Q$, we connect it with all vertices of all subgraphs $P_\icl$.
Second, we connect each vertex $p_{\icl,\ilit,\iend}$ with all vertices
$p_{\icl',\ilit',\iend'}$ for $\icl' \neq \icl$. Third, we connect each vertex $p_{\icl,\ilit,\iend}$ with all vertices in the sets $W$ in both copies of the gadget $H$,
except for $w^\icopy_{0,1}$ and $w^\icopy_{\ivar(\icl,\ilit),1-\ibit(\icl,\ilit)}$ for $\icopy\in \{1,2\}$.
This finishes the description of the clause gadgets.

Let us now describe the intuition behind this construction. In each gadget $P_\icl$ and in the guard subgraph $Q$ the edges are independent, thus
they need to be covered by different cliques. Two cliques are used to cover the edges of $Q$, and they can cover two out of three edges from each 
of the clause gadget $P_\icl$. The third one needs to be covered by the assignment clique twins from the gadgets $H$ (as the gadgets $P_\icl$ are not adjacent
to the vertices $u^\icopy_\ilnV$), and it corresponds to the choice which literal satisfies the clause $\clause_\icl$.
The missing edge $p_{\icl,\ilit,\iend}x_{0,1}$ ensures that only one clique of the assignment clique twins is used to cover the edges of $P_\icl$.
Finally, the missing edge $p_{\icl,\ilit,\iend}w^\icopy_{\ivar(\icl,\ilit),1-\ibit(\icl,\ilit)}$ verifies that this clique encodes a satisfying assignment of $\Phi$.

We note that it is non-trivial to cover the edges of $\Efree$ with $\Oh(\lnV) = \Oh(\log n)$ cliques induced by closed neighbourhoods of simplicial vertices. 
This is done in the next sections by appropriately using bit-vectors.

\subsubsection{Budget}

We set the number of cliques to cover the edges of $\Eimp$ as $k_0 = 2 \cdot 2 \cdot (\lnV-1) + 2 + 2 = 4\lnV$ --- two for each vertex $u_\ilnV^\icopy$, two
for the assignment clique twins in $H$, and two for the cliques that contain the edges of $Q$. The final number of cliques $k$ is the sum of $k_0$ and the number
of simplicial vertices introduced in the next section.

\subsubsection{Covering the free edges}

In this section we show that the edges of $\Efree$ can be covered by a small number of cliques, without accidentally covering any edge of $\Eimp$.
Moreover, such covering can be constructed in polynomial time. To construct the final \cliquecover{} instance, for each clique of this clique cover
we introduce a simplicial vertex adjacent to the vertex set of this clique, and raise the desired number of cliques by one.

\begin{lemma}\label{lem:cover-free}
The graph $\Gfree = (V(G),\Efree)$ admits a clique cover of size $46 + 36\lceil \log m \rceil + 24\lnV = \Oh(\log n)$. Moreover, such a clique cover can be constructed in polynomial time.
\end{lemma}
\begin{proof}
We start by noting that the free edges in the copies of the gadget $H$, and between these copies, can be covered by four cliques:
$\{w_{\ilnV,\ibit}^1: 0 \leq \ilnV < n\} \cup \{w_{\ilnV,\ibit'}^2: 0 \leq \ilnV < n\}$ for $(\ibit,\ibit') \in \{0,1\} \times \{0,1\}$.
Similarly, the edges incident to the guard gadget $Q$ can be covered with $24$ cliques:
$\{q_{\ilit,\iend}\} \cup \{p_{\icl,\ilit',\iend'}: 0 \leq \icl < m\}$ for $(\ilit,\iend,\ilit',\iend') \in \{1,2\} \times \{1,2\} \times \{1,2,3\} \times \{1,2\}$.

Covering the edges between different gadgets $P_\icl$ requires a bit more work.
For each $1 \leq \ilnV \leq \lceil \log m \rceil$ and $(\ilit,\iend,\ilit',\iend') \in (\{1,2,3\} \times \{1,2\})^2$ we take a clique $C_{\ilnV,\ilit,\iend,\ilit',\iend'}^P$
that contains exactly one vertex from each gadget $P_\icl$: if the $\ilnV$-th bit of the binary representation of $\icl$ equals $0$,
$p_{\icl,\ilit,\iend} \in C_{\ilnV,\ilit,\iend,\ilit',\iend'}^P$, and otherwise $p_{\icl,\ilit',\iend'} \in C_{\ilnV,\ilit,\iend,\ilit',\iend'}^P$.
Clearly, $C_{\ilnV,\ilit,\iend,\ilit',\iend'}^P$ induces a clique in $\Gfree$, as it contains exactly one vertex from each gadget $P_\icl$.
Let us now verify that all edges between the gadgets $P_\icl$ are covered by these $36\lceil \log m \rceil$ cliques.
Take any edge $p_{\icl,\ilit,\iend}p_{\icl',\ilit',\iend'} \in \Efree$, $\icl \neq \icl'$. Assume that the binary representations of $\icl$ and $\icl'$ differ
on the $\ilnV$-th bit; without loss of generality, assume that the $\ilnV$-th bit of $\icl$ is $0$, and the $\ilnV$-th bit of $\icl'$ is $1$.
Then the clique $C_{\ilnV,\ilit,\iend,\ilit',\iend'}^P$ contains both $p_{\icl,\ilit,\iend}$ and $p_{\icl',\ilit',\iend'}$.

We now handle the edges that connect the two copies of the gadget $H$ with the gadgets $P_\icl$.
First, we take care of the edges that are incident to the vertices $w_{0,0}^\icopy$. This can be easily done with $6$ cliques: for each $(\ilit,\iend) \in \{1,2,3\} \times \{1,2\}$
we take a clique that contains $w_{0,0}^1$, $w_{0,0}^2$ as well as all vertices $p_{\icl,\ilit,\iend}$ for $0 \leq \icl < m$.
Second, we take care of the edges $p_{\icl,\ilit,\iend}w_{\ivar(\icl,\ilit),\ibit(\icl,\ilit)}^\icopy$. To this end, we take $12$ cliques: for each $(\ilit,\iend,\ibit) \in \{1,2,3\} \times \{1,2\} \times \{0,1\}$ we take a clique that contains $w_{\ivar,\ibit}^\icopy$ for $1 \leq \ivar < n$, $\icopy \in \{1,2\}$ as well as all vertices $p_{\icl,\ilit,\iend}$ that satisfy
$\ibit(\icl,\ilit) = \ibit$.

We are left with the edges of form $p_{\icl,\ilit,\iend}w_{\ivar,\ibit}^\icopy$ for $\ivar \notin \{0,\ivar(\icl,\ilit)\}$.
These edges can be covered in a similar fashion to the edges between the gadgets $P_\icl$.
For each $1 \leq \ilnV \leq \lnV$ and $(\ilit,\iend,\ibit,\ibit') \in \{1,2,3\} \times \{1,2\} \times \{0,1\}^2$ we construct a clique
$C_{\ilnV,\ilit,\iend,\ibit,\ibit'}^W$ that contains all vertices $w_{\ivar,\ibit}^\icopy$ for $\icopy \in \{1,2\}$ and $1 \leq \ivar < n$ such that
the $\ilnV$-th bit of the binary representation of $\ivar$ equals $\ibit'$, as well as all vertices $p_{\icl,\ilit,\iend}$ for $0 \leq \icl < m$ such that
the $\ilnV$-th bit of the binary representation of $\ivar(\icl,\ilit)$ equals $1-\ibit'$.
To see that $C_{\ilnV,\ilit,\iend,\ibit,\ibit'}^W$ is indeed a clique in $\Gfree$, note that
it contains only edges in $G[W_0^1 \cup W_0^2]$ or $G[W_1^1 \cup W_1^2]$, between different gadgets $P_\icl$, and edges of the form $p_{\icl,\ilit,\iend}w_{\ivar,\ibit}^\icopy$ where
$\ivar \neq 0$ and $\ivar \neq \ivar(\icl,\ilit)$ (the indices $\ivar$ and $\ivar(\icl,\ilit)$ must differ on the $\ilnV$-th bit in the clique
  $C_{\ilnV,\ilit,\iend,\ibit,\ibit'}^W$).
We finish the proof of the lemma by verifying that all edges of the form $p_{\icl,\ilit,\iend}w_{\ivar,\ibit}^\icopy$ for $\ivar \notin \{0,\ivar(\icl,\ilit)\}$
are covered by these $24\lnV$ cliques. As $\ivar \neq \ivar(\icl,\ilit)$, there exists $1 \leq \ilnV \leq \lnV$ such that $\ivar$
and $\ivar(\icl,\ilit)$ differ on the $\ilnV$-th bit of their binary representations. Let $\ibit'$ be the $\ilnV$-th bit of the binary representation
of $\ivar$. We infer that both $p_{\icl,\ilit,\iend}$ and $w_{\ivar,\ibit}^\icopy$ are included in the clique
$C_{\ilnV,\ilit,\iend,\ibit,\ibit'}$ and the lemma is proven.
\end{proof}

Recall that for each clique constructed by Lemma \ref{lem:cover-free} we add a simplicial vertex to $G$ that is adjacent to all vertices of this clique.
The simplicial vertices are independent in $G$. As discussed earlier, we can assume that for any non-isolated simplicial vertex $s$ in $G$,
any optimal clique cover in $G$ contains a clique whose vertex set equals to the closed neighbourhood of $s$.

We conclude the construction section by setting the desired number of cliques $k$ to be the sum of $k_0$ and the number of aforementioned simplicial vertices,
$k = 4\lnV + 46 + 36\lceil \log m \rceil + 24\lnV = \Oh(\log n)$.

\subsection{Completeness}\label{ss:eth-completeness}

In this section we show how to translate a satisfying assignment of the input formula $\Phi$ into a clique cover of $G$ of size $k$.

\begin{lemma}\label{lem:eth-completeness}
If the input formula $\Phi$ is satisfiable, then there exists a clique cover of the graph $G$ of size $k$.
\end{lemma}
\begin{proof}
Let $\phi: \{0,1,\ldots,n-1\} \to \{0,1\}$ be a satisfying assignment of $\Phi$, that is,
    $\phi(\ivar)$ is the value of $x_\ivar$, $0$ stands for false and $1$ stands for true.
By the properties of the preprocessing step of the construction, we may assume that $|\phi^{-1}(0)| = |\phi^{-1}(1)| = |\Vars|/2 = 2^{\ell-1}$
and that $\phi(0) = 0$ (as $x_0$ is a dummy variable).

We start the construction of the clique cover $\ccfam$ of the graph $G$ by taking into $\ccfam$, for each of the $46 + 36\lceil \log m \rceil + 24\lnV$ simplicial vertices
of $G$ constructed in Lemma~\ref{lem:cover-free}, a clique induced by the closed neighbourhood of the simplicial vertex. In this manner we cover
all edges of $\Efree$, and we are left with a budget of $4\lnV$ cliques.

We define the assignment clique twins $C_0^A$ and $C_1^A$. For each clause $\clause_\icl$ of $\Phi$, let $\ilit(\icl)$ be an index of a literal
that is satisfied by $\phi$ in $\clause_\icl$ (if there is more than one such literal, we choose an arbitrary one). The clique $C_0^A$ contains
the vertices $w_{\ivar,\phi(\ivar)}^\icopy$ for $0 \leq \ivar < n$ and $\icopy \in \{1,2\}$ as well as the following vertices from the clause gadgets:
$p_{\icl,\ilit(\icl),\iend}$ for $0 \leq \icl < m$, $\iend \in \{1,2\}$. Note that $C_0^A$ is indeed a clique,
since the only missing edges between vertices $w_{\ivar,\ibit}^\icopy$ and $p_{\icl,\ilit,\iend}$ are either incident to $w_{0,1}^\icopy$
(but $\phi(0) = 0$) or of the form $w^\icopy_{\ivar(\icl,\ilit),1-\ibit(\icl,\ilit)} p_{\icl,\ilit,\iend}$ (but
$\phi(\ivar(\icl,\ilit(\icl)))$ satisfies $\ilit(\icl)$-th literal of $\clause_\icl$, i.e., $\ibit(\icl,\ilit(\icl))=\phi(\ivar(\icl,\ilit(\icl)))$).

The clique $C_1^A$ is the twin (complement) of the clique $C_0^A$ in both copies of the graph $H_{\lnV+1}$, i.e., $C_1^A = \{w_{\ivar,1-\phi(\ivar)}^\icopy: 0 \leq \ivar < n, \icopy \in \{1,2\}\}$.
Clearly, $C_1^A$ is a clique in $G$.

Let us now fix $\icopy \in \{1,2\}$ and focus on the graph $G[W^\icopy]$ isomorphic to $H_{\lnV+1}$. The edges of $\Efree$ in this subgraph form clique twins
$\{w_{\ivar,c}^\icopy : 0 \leq \ivar < n\}$ for $c \in \{0,1\}$.
The assignment clique twins $C_0^A$ and $C_1^A$ form second clique twins in $G[W^\icopy]$, after truncating them to this subgraph.
Moreover, the assumption that $\phi$ evaluates exactly half of the variables to false and half to true implies that these two clique twins satisfy the assumptions
of Lemma \ref{lem:extend-twin}. We infer that all remaining edges of $G[W^\icopy]$ can be covered by $\lnV-1$ clique twins; we add the vertex $u_\ilnV^\icopy$
to both cliques of the $\ilnV$-th clique twin, and add these clique twins to the constructed clique cover $\ccfam$.
In this manner we cover all edges incident to all vertices $u_\ilnV^\icopy$ for $1 \leq \ilnV < \lnV$, $\icopy \in \{1,2\}$.
As we perform this construction for both values
$\icopy \in \{1,2\}$, we use $4\lnV-4$ cliques, and we are left with a budget of two cliques.

The cliques introduced in the previous paragraph cover all edges in $\Eimp$ in both copies of the assignment gadget $H$. We are left with the clause gadgets $P_\icl$
and the guard gadget $Q$. Recall that the clique $C_0^A$ covers one out of three edges in each gadget $P_\icl$. Thus it is straightforward to cover the remaining edges
with two cliques: each clique contains both endpoints of exactly one uncovered edge from each gadget $P_\icl$ and the gadget $Q$. This finishes the proof of the lemma.
\end{proof}

\subsection{Soundness}\label{ss:eth-soundness}

In this section we show a reverse transformation: a clique cover of $G$ of size at most $k$ cannot differ much from the one constructed in the proof of Lemma \ref{lem:eth-completeness}
and, therefore, encodes a satisfying assignment of the input formula $\Phi$.

\begin{lemma}\label{lem:eth-soundness}
If there exists a clique cover of $G$ of size at most $k$, then the input formula $\Phi$ is satisfiable.
\end{lemma}
\begin{proof}
Let $\ccfam$ be a clique cover of size at most $k$ of $G$. As $G$ contains $k-4\lnV$ simplicial vertices, without loss of generality we may assume
that, for each such simplicial vertex $s$, the family $\ccfam$ contains a clique induced by the closed neighbourhood of $s$.
These cliques cover the edges of $\Efree$, but no edge of $\Eimp$. Let $\ccfam_0 \subseteq \ccfam$ be the set of the remaining cliques; $|\ccfam_0| \leq 4\lnV$.

Let us start with analyzing the guard gadget $Q$. It contains two independent edges from $\Eimp$. Thus, $\ccfam_0$ contains two cliques, each containing one edge of $Q$.
Denote this cliques by $C_1^Q$ and $C_2^Q$. 
Note that each clause gadget $P_\icl$ contains three independent edges, and only two of them may be covered by the cliques $C_1^Q$ and $C_2^Q$. Thus there exists
at least one additional clique in $\ccfam_0$ that contains an edge of at least one gadget $P_\icl$; let us denote this clique by $C_0^A$ (if there is more than one such clique,
we choose an arbitrary one).

Each vertex $u_\ilnV^\icopy$ for $1 \leq \ilnV < \lnV$, $\icopy \in \{1,2\}$ needs to be contained in at least two cliques of $\ccfam_0$, since all its incident edges 
are in $\Eimp$ and the neighbourhood of $u_\ilnV^\icopy$ is not a clique. Moreover,
no vertex $u_\ilnV^\icopy$ may belong to $C_1^Q$, $C_2^Q$ nor to $C_0^A$, as these vertices are not adjacent to the vertices of $P_\icl$ and $Q$.
As there are $2\lnV-2$ vertices $u_\ilnV^\icopy$, the vertices $u_\ilnV^\icopy$ are independent, and $|\ccfam_0 \setminus \{C_1^Q,C_2^Q,C_0^A\}| \leq 4\lnV-3$,
we infer that at most one vertex $u_\ilnV^\icopy$ may be contained in more than two cliques from $\ccfam_0$. Without loss of generality we can assume that this
vertex belongs to the second copy of the assignment gadget $H$, that is, each vertex $u_\ilnV^1$ for $1 \leq \ilnV < \lnV$ belongs to exactly two cliques $C_{\ilnV,0}^U, C_{\ilnV,1}^U \in \ccfam_0$.

Note that the only way to cover the edges incident to $u_\ilnV^1$ with only two cliques $C_{\ilnV,0}^U$, $C_{\ilnV,1}^U$ is to take these cliques to induce clique twins
in $G[W^1] \cong H_{\lnV+1}$. That is, $C_{\ilnV,0}^U$ consists of $u_\ilnV^1$ and exactly one endpoint of each non-edge of $G[W^1]$,
and $C_{\ilnV,1}^U = \{u_\ilnV^1\} \cup (W^1 \setminus C_{\ilnV,0}^U)$.

We infer that the cliques $\{C_{\ilnV,\ibit}^U: 1 \leq \ilnV < \lnV, \ibit \in \{0,1\}\}$,
together with the clique twins formed by the edges of $\Efree$ in $G[W^1]$, sum up to $\lnV$ clique twins in $G[W^1] \cong H_{\lnV+1}$.
Moreover, the edges of $G[W^1]$ may be covered with only two additional cliques (including $C_0^A$): there are at least $4l-4$ cliques that contain vertices $u_\ilnV^1$ or $u_\ilnV^2$, out of which exactly $2l-2$ can contain vertices from $W^1$, while at least two other cliques have to contain vertices of $Q$, thus having to be disjoint with $W^1$.
Lemma \ref{lem:last-twin} asserts that the only way to cover the edges of $G[W^1] \cong H_{\lnV+1}$ is to use one additional pair of cliques that form a clique twin; thus,
$C_0^A \cap W^1$ is a maximum clique in $G[W^1]$ and $\ccfam_0$ contains a clique $C_1^A$ such that $C_1^A \cap W^1 = W^1 \setminus C_0^A$.

Recall that the clique $C_0^A$ contained an edge from at least one gadget $P_\icl$. Therefore, $x_{0,1}^1 \notin C_0^A$, as $x_{0,1}^1$ is not adjacent to any vertex in any gadget $P_\icl$.
Since $C_0^A$ and $C_1^A$ induce clique twins in $G[W_1]$, we infer that $x_{0,1}^1 \in C_1^A$ and $C_1^A$ is disjoint with all gadgets $P_\icl$.
As the vertices $u_\ilnV^\icopy$ are not adjacent to the gadgets $P_\icl$, the cliques that cover the edges incident to the vertices $u_\ilnV^\icopy$
cannot cover any edges of the gadgets $P_\icl$ either. We infer that the edges of the gadgets $P_\icl$ are covered by only three cliques --- $C_0^A$, $C_1^Q$ and $C_2^Q$ ---
and that in each gadget $P_\icl$, each of this three cliques contains exactly one edge. For a clause $\clause_\icl$, let $\ilit(\icl)$ be the index of the literal whose
edge is covered by $C_0^A$.

We claim that an assignment $\phi:\{0,1,\ldots,n-1\} \to \{0,1\}$ that assigns the value $\phi(\ivar)$ to the variable $x_\ivar$ in such a manner
that $w_{\ivar,\phi(\ivar)}^1 \in C_0^A$, satisfies the formula $\Phi$. More precisely, we claim that
for each clause $\clause_\icl$, the $\ilit(\icl)$-th literal of $\clause_\icl$
satisfies this clause in the assignment $\phi$.
Indeed, as $p_{\icl,\ilit(\icl),\iend} \in C_0^A$ for $\iend \in \{1,2\}$, and $p_{\icl,\ilit(\icl),\iend}$ is not adjacent to $w_{\ivar(\icl,\ilit(\icl)),1-\ibit(\icl,\ilit(\icl))}^1$,
  we have that $w_{\ivar(\icl,\ilit(\icl)),\ibit(\icl,\ilit(\icl))}^1 \in C_0^A$ and $\phi(\ivar(\icl,\ilit(\icl))) = \ibit(\icl,\ilit(\icl))$.
This finishes the proof of the lemma and concludes the proof of Theorem \ref{thm:eth-reduction}.
\end{proof}

\section{Conclusions}
In this paper we have shown a double-exponential lower bound for \cliquecover{}
parameterized by the number of cliques, obtaining tight complexity bounds
for this problem. 

Prior to this work, several matching lower bounds for $2^{O(k)}\poly(n)$, $2^{O(\sqrt{k})}\poly(n)$
and $2^{O(k \log k)}\poly(n)$ time algorithms for parameterized problems
were already known. 
However, to the best of our knowledge \cliquecover{} is the first
example of a natural parameterization for which double-exponential
upper and lower bounds are proved.
We hope our work will inspire further results of this type.

\bibliographystyle{abbrv}
\bibliography{edge-clique-cover}

\end{document}